\documentclass[twoside]{article}
\usepackage{qic,epsfig}
\usepackage{shadow,epsf,amssymb,amsmath}
\usepackage{cite}
\usepackage{booktabs}

\usepackage{multirow}
\usepackage{graphicx}
\usepackage{rotating}
\usepackage{enumerate}
\usepackage{subfigure}
\usepackage[pdftex]{color}
\usepackage[ruled,vlined,linesnumbered]{algorithm2e}
\newcommand{\ignore}[1]{}
\input{Qcircuit.tex}

\newcommand{\BDD}{\text{BDD}}

\newcommand{\dist}{\mathrm{dist}}

\newcommand{\Filter}{\mathrm{Filter}}

\newtheorem{app-lemmaenv}[section]{Lemma}

\newcommand{\bmark}[1]{{\color{blue} #1}}
\newcommand{\bB}{\textbf{B}}

\newcommand{\cO}{\mathcal{O}}
\newcommand{\cG}{\mathcal{G}}

\newcommand{\Tof}{\textrm{Toffoli}}
\newcommand{\enum}{\textsf{Enum}}
\newcommand{\enump}{\textsf{EnumP}}
\newcommand{\qenump}{\textsf{QEnumP}}
\newcommand{\qsvp}{\textsf{QSVP}}
\newcommand{\poly}{\textrm{poly}}
\newcommand{\bdd}{\textsf{BDD}}
\newcommand{\bddp}{\textsf{BDDP}}
\newcommand{\bdda}{\textsf{BDD}_{\alpha}}

\newcommand{\BDDP}{\text{BDDP}}

\begin{document}
\setlength{\textheight}{8.0truein}    
\runninghead{Space-efficient Classical and Quantum Algorithms for the Shortest Vector Problem}
            {Y. Chen, K.-M.Chung, and C.-Y. Lai}

\normalsize\textlineskip
\thispagestyle{empty}
\setcounter{page}{1}


\vspace*{0.88truein}

\alphfootnote

\fpage{1}
\centerline{\bf
SPACE-EFFICIENT CLASSICAL AND QUANTUM}
\vspace*{0.035truein}
\centerline{\bf ALGORITHMS FOR THE SHORTEST VECTOR PROBLEM }
\vspace*{0.37truein}
\centerline{\footnotesize
YANLIN CHEN}
\vspace*{0.015truein}
\centerline{\footnotesize\it Institute of Information Science, Academia Sinica}
\baselineskip=10pt
\centerline{\footnotesize\it  Taipei, Taiwan.}
\vspace*{10pt}
\centerline{\footnotesize 
KAI-MIN CHUNG}
\vspace*{0.015truein}
\centerline{\footnotesize\it Institute of Information Science, Academia Sinica }
\baselineskip=10pt
\centerline{\footnotesize\it Taipei, Taiwan.}
\vspace*{10pt}
\centerline{\footnotesize 
CHING-YI LAI}
\vspace*{0.015truein}
\centerline{\footnotesize\it Institute of Information Science, Academia Sinica }
\baselineskip=10pt
\centerline{\footnotesize\it  Taipei, Taiwan.}
\vspace*{0.225truein}

\vspace*{0.21truein}
\abstracts{
A lattice is the integer span of some linearly independent vectors. Lattice problems have many significant applications in coding theory and cryptographic systems for their conjectured hardness. The Shortest Vector Problem (SVP), which asks to find a shortest nonzero vector in a lattice, is one of the well-known problems that are believed to be hard to solve, even with a quantum computer.\\
 In this paper we propose  space-efficient  classical  and  quantum algorithms for solving SVP. Currently the  best time-efficient algorithm for solving SVP takes $2^{n+o(n)}$ time and $2^{n+o(n)}$ space. Our classical algorithm takes $2^{2.05n+o(n)}$ time to solve SVP and it requires only $2^{0.5n+o(n)}$ space. We then adapt our classical algorithm to a quantum version, which can solve SVP in time $2^{1.2553n+o(n)}$  with $2^{0.5n+o(n)}$ classical space and only poly$(n)$ qubits.}{}{}

\vspace*{10pt}
\vspace*{3pt}

\vspace*{1pt}\textlineskip	
\section{Introduction}	        
\vspace*{-0.5pt}
\noindent
\label{Introduction}
Quantum attackers refer to those who own the power of quantum computation and are malicious to a cryptographic system. Since they are capable of breaking  some existing cryptographic systems~\cite{shor1999polynomial}, computer scientists aim to find cryptographic systems that are secure against the threat of quantum attackers.
Such a cryptographic system is referred to post-quantum cryptography 
and one well-known example is \emph{lattice}-based cryptography.
%

A lattice  is the set of integer combinations of some linearly independent (basis) vectors
and it is related to many topics in pure mathematics and applied mathematics, especially in number theory, group theory, convex optimization, cryptography and coding theory~\cite{minkowski1891ueber}.
As for cryptographic systems, it was first shown in Ajtai's seminal report~\cite{ajtai1996generating} that one can  build cryptographic primitives with lattices.
In~\cite{micciancio2007worst}, Regev and Micciancio showed that finding small integer solutions to certain random modular linear equations
is at least as hard as solving certain hardest lattice problems.
Later Regev proved a reduction from certain hardest lattice problems to learning with errors problem with quantum computation power
and he proposed a classical public-key cryptosystem whose security is based on the hardness of the learning problem~\cite{Regev:2005:LLE:1060590.1060603}.
In addition, Regev and Micciancio introduced  efficient cryptographic systems based on the conjectured intractability of solving lattice problems~\cite{micciancio2009lattice}.
Furthermore, Gentry built a fully homomorphism encryption system based on the hardness of a lattice problem~\cite{gentry2009fully}.
There are  also many other cryptographic systems     related to lattice problems.

 Since the security of a lattice-based cryptography depends on the hardness of solving a lattice problem, it is important to determine the existence of an efficient algorithm for  this problem.  One of the hard lattice problems  is the \emph{shortest vector problem} (SVP):  given a lattice $L\subset \mathbb{R}^n$ (of full rank), the goal is to output a shortest nonzero vector in $L$ and its length is denoted by $\lambda_1(L)$.
(In the following $n$ is referred to the rank of the lattice.)
{ The approximation version of SVP (called \emph{$\gamma$-SVP}) asks to find a nonzero vector in $L$ with length at most $\gamma\cdot \lambda_1(L)$, where $\gamma$ is a function of  $n$.
In practical, if there exists an efficient algorithm that solves the  $\gamma$-approximation SVP with $\gamma=$poly($n$), those cryptographic primitives based on the hardness of SVP would be insecure.}
 Hence we would like to know how fast we can solve SVP by a classical or quantum algorithm.

{To solve $\gamma$-SVP with an almost constant approximation factor  was shown to be NP-hard (under randomized reductions)~\cite{Dinur:2000:ASW:648257.752906, Ajtai:1998:SVP:276698.276705, Khot:2005:HAS:1089023.1089027, Haviv:2007:THS:1250790.1250859, Micciancio:2001:SVL:586841.586903}.
 However,  when the approximation factor is above $n^2$, it is probably easier than all NP-complete problems~\cite{aharonov2005lattice}.}
 When the approximation factor  is  exponential in $n$, Lenstra \emph{et al.} proposed a method (the celebrated Lenstra-Lenstra-Lov\'{a}sz (LLL) algorithm) to construct a polynomial solver to this $\gamma$-SVP problem~\cite{lenstra1982factoring}.
 Soon after, Schnorr gave an 
  algorithm to solve $\gamma$-SVP (with approximation factor $r^{n/r}$) in $2^{O(r)}\text{poly}(n)$ time for any $r\geq 2$~\cite{schnorr1994lattice}. Then Buchmann and Ludwig improved Schnorr's algorithm with better parameters~\cite{buchmann2006practical}.

On the other hand, Ajtai, Kumar, and Sivakumar gave a randomized algorithm with both time and space complexity $2^{O(n)}$ (the so-called AKS sieving algorithm).
 { This time complexity was determined by Regev to be $2^{16n+o(n)}$~\cite{Regev:2004lecture} and then improved to $2^{5.9n+o(n)}$ by Nguyen and Vidick~\cite{nguyen2008sieve}. The time complexity of AKS sieve was further improved to $2^{2.465n+o(n)}$  by Pujol and Stehle~\cite{cryptoeprint:2009:605}. }
Later  Micciancio and Voulgaris\cite{micciancio2013deterministic} gave a deterministic classical algorithm that solves exact SVP and some other lattice problems in the  worse case in  time $2^{2n+(n)}$ and  space $2^{n+o(n)}$.
 Recently Aggarwal \emph{et al.} provided a classical probabilistic algorithm that solves exact SVP  with probability $1-e^{-\Omega(n)}$
  and it takes time $2^{n+o(n)}$ and space $2^{n+o(n)}$~\cite{aggarwal2015solving}.

 In fact we can solve SVP much faster under various heuristic assumptions.  Nguyen and Vidick gave a heuristic variant of the AKS sieving  algorithm that requires time $2^{0.415n+o(n)}$   and space $2^{0.208n+o(n)}$~\cite{nguyen2008sieve}. Later Micciancio and Voulgaris proposed a different  type of heuristic sieve algorithm, called GaussSieve,  and it performs very well in practice~\cite{Micciancio:2010:FET:1873601.1873720}.  { Recently, Laarhoven \emph{et al.} developed various classical and quantum sieve algorithms with time-space tradeoff~\cite{laarhoven2015finding}~\cite{laarhoven2015search}~\cite{bai2016tuple}~\cite{Becker:2016:NDN:2884435.2884437}~\cite{HKL17tradeoff}.}

It has been shown that quantum algorithms can have speedup over classical ones (e.g. Shor's factoring algorithm~\cite{shor1999polynomial} and Grover's search algorithm \cite{Grover1996rk}).
 Nevertheless, it seems that quantum power may not be much helpful in time for solving SVP, since most of the classical algorithms for solving lattice problems are recursive ones, which are not likely to
 have speedup from quantum parallel computation.
To date the most efficient quantum algorithm for solving SVP is provided by Laarhoven \emph{et al.} and it takes time $2^{1.799n+o(n)}$.
In addition, they also proposed a heuristic quantum algorithm that can solve SVP in $2^{0.265n+o(n)}$~\cite{laarhoven2015finding}).
In this paper, we show that quantum power does help us to solve lattice problem space-efficiently.
We will provide  a quantum algorithm that solves SVP in time $2^{1.2553n+o(n)}$ and it requires  classical space $2^{0.5n+o(n)}$ and only polynomially many qubits.
Along the way,  we also introduce a classical algorithm for solving SVP in time $2^{2.05n+o(n)}$ and space $2^{0.5n+o(n)}$.
Both algorithms are mainly built on a lattice  \emph{enumeration} algorithm $\enum$.

Lattice enumeration is a standard technique to solve SVP by systematically enumerating all lattice points in a bounded region of space.
Many enumeration algorithm can solve SVP with  only polynomial space. However, they usually run in $n^{O(n)}$ time~\cite{Micciancio:2015:FLP:2722129.2722150, Kannan:1987:MCB:35577.35580, hanrot2007improved} and  are not comparable to other algorithms that can solve SVP in single exponential time. In contrast, Kirchner and Fouque proposed a lattice enumeration algorithm with tradeoff between  time and space~\cite{kirchner2016time}, which was then used to construct an algorithm that solves SVP and both of its time and space complexity are $3^{n+o(n)}$.
 (In the extreme case, it can solve SVP in time $O(n^{n/4})$ with only polynomial space.)

Inspired by the enumeration method in~\cite{kirchner2016time}, we construct a space-efficient enumeration algorithm $\enum$ by using a \emph{BDD oracle}.
Roughly speaking,  for a target vector close enough to a lattice $L$, the search problem $\alpha$-BDD (Bounded Distance Decoding Problem)
is to find  a lattice point $y\in L$ such that $\|y-t\| \leq \alpha \cdot \lambda_1(L)$ for $\alpha<0.5$.
A \emph{BDD oracle} is an oracle that solves BDD. 
Our classical algorithm $\enum$ makes a list of the lattice points within a specified distance (called \emph{enumeration radius}) to a target vector
by using a BDD oracle:
\begin{theorem}\label{Enum}(informal)
Given a lattice $L \subset \mathbb{R}^n$, a target vector $t\in\mathbb{R}^n$, an $\alpha$-$\BDD$ oracle $\bdda$ with $\alpha < 0.5$, and an integer scalar $p$ such that $p\alpha>1$,  there exists a classical algorithm   that collects all lattice points within distance $p\alpha \lambda_1(L)$ to $t$ by querying  $\bdda$ oracle $p^n$ times.
\end{theorem}

$\enum$ provides a tradeoff between the enumeration radius and time complexity. One can solve SVP by $\lceil\frac{1}{\alpha}\rceil^n$ queries to the $\bdda$ oracle.
Note that the $\bdda$ oracle does not work for $\alpha\geq 0.5$, and therefore  $p$ has to be at least $3$. 
In addition, a central component of an $\alpha$-BDD algorithm is the preparation of a discrete Gaussian distribution.
Observe that the same discrete Gaussian samples can be used whenever a BDD oracle is queried.
Consequently,  these discrete Gaussian samples can be prepared in advance and hence save the overall time complexity.
Moreover,  Aggarwal \emph{et al.} have discussed how to prepare these discrete Gaussian samples~\cite{aggarwal2015solving}
and we will adopt their method.
With this preprocessing, the resulting BDD problem is called a BDDP problem.
 (For the formal definition of BDDP, please refer to definition~\ref{BDDP}.)
It remains to  build a $\bddp$ oracle. Then we have the following theorem:
  \begin{theorem}\label{C_SVPsolver}
  (informal)
There exists a classical probabilistic algorithm that solves SVP with probability $1-2^{-\Omega(n)}$ in time $2^{2.0478n+o(n)}$ and in space $2^{n/2+o(n)}$.
\end{theorem}

Next we use the enumeration algorithm $\enum$ as a quantum subroutine, and then apply the idea of quantum Grover search to amplify the probability of finding the correct answer.
Since $\enum$  makes a finite list of candidates for SVP, we can use Grover search to find one with minimum length in the list and have potential quantum speedup.
 \begin{theorem}[Main Theorem]
There is a quantum algorithm that solves SVP with probability $1-e^{\Omega(n)}$ that requires  $2^{1.2553n+o(n)}$ elementary quantum gates,  $2^{0.5n+o(n)}$  classical space, and only poly$(n)$  qubits.
\end{theorem}

{ Compared to other single exponential time quantum algorithms, our quantum algorithm $\qsvp$ for solving SVP needs only polynomially many qubits and it does not need quantum RAM model for accessing classical memories. In addition, $\qsvp$ uses exactly $2^{0.5n+o(n)}$ classical space and only poly($n$)  qubits. We remark that $\qsvp$ is the first single exponential time quantum algorithm that solves SVP with only polynomially many qubits.}
In Table~\ref{tb:SVP_result}, we list some heuristic or provable algorithms for solving SVP in the past two decades.  One can  find that both our classical algorithm $\enump$ and quantum algorithm $\qsvp$ use the least space compared to other  provable algorithms with a single exponential time complexity.

\begin{table}[h]
    \centering
    \begin{tabular}{| l | l | l | l | l |}
    \hline
    Algorithm & {$\text{Approximation}\atop \text{factor}$} & Time complexity  & Space complexity &  Type\\ \hline
    Sch87~\cite{schnorr1994lattice}    &  $r^{n/r}$  & $2^{O(r)}\cdot \text{poly}(n)$ & $\text{poly}(n)$  & classical, heuristic \\ \hline
    Lud03~\cite{ludwig2003faster}    & \tiny$(0.167r)^{0.5n/r} $   & \tiny	$O((0.167r)^{0.125r})\cdot \text{poly}(n)$ & $\normalsize\text{poly}(n)$  & quantum, heuristic \\ \hline
    NV08~\cite{nguyen2008sieve}  & $1 $ & $2^{0.415n+o(n)}$   & $2^{0.208n+o(n)}$ & classical, heuristic \\ \hline
    LMP15~\cite{laarhoven2015finding}    & 1           & $2^{0.312n+o(n)}$      &$2^{0.208n+o(n)}$, QRAM     & quantum, heuristic   \\ \hline
    LMP15~\cite{laarhoven2015finding}    & 1           & $2^{0.268n+o(n)} $       &$2^{0.268n+o(n)}$, QRAM     & quantum, heuristic   \\ \hline
        Laa15~\cite{laarhoven2015search}    & 1           & $2^{0.265n+o(n)}$        &$2^{0.265n+o(n)}$, QRAM     & quantum, heuristic   \\ \hline
    BLS16~\cite{bai2016tuple}  & 1           & $2^{0.4812n+o(n)}$       &$2^{0.1887n+o(n)}$         & classical, heuristic \\ \hline
    BDGL16~\cite{Becker:2016:NDN:2884435.2884437}  & 1           & $2^{0.2925n+o(n)}$       &$2^{0.208n+o(n)}$        & classical, heuristic \\ \hline
    LLL82\cite{lenstra1982factoring}    & $2^{0.5n}$   & $\text{poly}(n)$    & $\text{poly}(n)$           & classical, provable \\ \hline
    Kan83~\cite{Kannan:1983:IAI:800061.808749}    & $1 $   & $n^{n+o(n)}$ & $\text{poly}(n)$  & classical, provable \\ \hline
    Hel85~\cite{Helfrich:1985:ACM:6566.6567}    & $1 $   & $n^{0.5n+o(n)}$ & $\text{poly}(n)$  & classical, provable \\ \hline
    AKS01~\cite{Ajtai:2001:SAS:380752.380857}    & $1 $   & $2^{O(n)}$ & $2^{O(n)}$  & classical, provable \\ \hline
    Reg04~\cite{Regev:2004lecture}    & $1 $   & $2^{16n+o(n)}$ & $2^{8n+o(n)}$  & classical, provable \\ \hline
    HS07~\cite{DBLP:journals/corr/abs-0705-0965}    & $1 $   & $n^{0.184n+o(n)}$ & $\text{poly}(n)$  & classical, provable \\ \hline
    NV08~\cite{nguyen2008sieve}    & $1 $   & $2^{5.90n+o(n)}$ & $2^{2.95n+o(n)}$  & classical, provable \\ \hline
    PS09~\cite{cryptoeprint:2009:605}    & 1           & $2^{2.465n+o(n)}$         & $2^{1.233n+o(n)}$        & classical, provable \\ \hline
        MV09~\cite{Micciancio:2010:FET:1873601.1873720}    & 1    & unknown       &  $2^{0.41n+o(n)}$        & classical, provable \\ \hline
  MV10~\cite{micciancio2013deterministic}    & 1           & $2^{2n+o(n)}$         & $2^{n+o(n)}$        & classical, provable \\ \hline
    LMP15~\cite{laarhoven2015finding}    & 1           & $2^{1.799n+o(n)}$      &$2^{1.286n+o(n)}$, QRAM     & quantum, provable   \\ \hline
    ADRS15~\cite{aggarwal2015solving}  & 1           & $2^{n+o(n)}$       &$2^{n+o(n)}$          & classical, provable \\ \hline
     $\enump$  & 1           & $2^{2.05n+o(n)}$       &$2^{n/2+o(n)}$          & classical, provable \\ \hline
      $\qsvp$& $1$      & $2^{1.2553n+o(n)}$    &$2^{n/2+o(n)}$         & quantum, provable   \\
    \hline
    \end{tabular}
    \tcaption{{ Known algorithms for solving the shortest vector problem. Note that in LMP15~\cite{laarhoven2015finding} and Laa15~\cite{laarhoven2015search}, they used quantum RAM model (QRAM), or RAM-like quantumly addressable classical memories for the quantum search algorithms, while our quantum algorithm $\qsvp$ only needs polynomially many qubits and $2^{0.5n+o(n)}$ classical space. MV09~\cite{Micciancio:2010:FET:1873601.1873720} experimentally takes $2^{0.415n+o(n)}$ time and  $2^{0.208n+o(n)}$ spaces.} }
    \label{tb:SVP_result}
\end{table}

\section{Preliminaries}
In this paper, the notation $\log$ is the natural logarithm and $\log_2$ is the base-2 logarithm.

\subsection{Lattice}
First we introduce the notation in this paper. Suppose $\textbf{B}= \left\{{\vec{b}_1},\dots,{\vec{b}_n}\right\}$, $n\leq m$, is a set of independent vectors in $\mathbb{R}^m$, where ${\vec{b}_j}$ are considered as column vectors. The lattice space generated by $\bB$ is
\[
L=\left\{\sum_i c_i{\vec{b}_i}\mid  c_i \in \mathbb{Z}\right\},
\]
and $\bB$ is called a \emph{basis} of $L$. In other words, $L$ is the integer span of the basis $\bB$.
Equivalently,
\[
L=\left\{B\vec{x}\mid \vec{x}\in \mathbb{Z}^n \right\},
\]
where   $\mathrm{B}=[{\vec{b}_1}\ {\vec{b}_2}\cdots {\vec{b}_n}]\in\mathbb{R}^{m\times n}$ is called a \emph{basis matrix} of $L$.
Here we only consider lattice bases of full rank, that is, $n=m$. Note that a lattice may be generated by different bases. A simple example is that  $\{(0,1)^T, (1,0)^T\}$ and $\{(100,1)^T, (101,1)^T\}$ generate
the same lattice space  $\{(a,b): a,b\in \mathbb{Z}\}\subset \mathbb{R}^2$.
Thus we may write $L(\bB)$ to indicate that $\bB$ is a basis of $L$.

An element $v$ in $L$ is called a \emph{lattice point},
and its length is $\|v\|$, where $\|\cdot\|$ is the $l_2$ norm in $\mathbb{R}^n$.
For $x, y \in \mathbb{R}^n$, we define an equivalence relation $$y = x \mod L$$ if $y-x \in L$.
For  $t \in \mathbb{R}^n$, the distance between $t$ and $L$ is defined as  $$\dist(t,L)=\min\limits_{x\in L}\|t-x\|.$$
The following definition will be used in the proofs of our results.
\begin{definition}\label{sucMin}
For a lattice $L \subset \mathbb{R}^n$, the $i$th successive minimum of $L$ is
\[
\lambda_i(L) \equiv \inf \{r:\mathrm{dim}(\mathrm{span}(L\cap \mathrm{Ball}(0,r)))\geq i\},
\]
where $\mathrm{Ball}(0,r)$ denotes a closed ball with center at the origin and radius $r$.
\end{definition}
\noindent In particular, $\lambda_1(L)$ is the length of the shortest (nonzero) vector in $L$.

A scaled lattice space $pL$ for some integer $p>1$ is defined as
\begin{align}
pL\equiv \left\{p \cdot \sum_i c_i{\vec{b}_i}\mid  c_i \in \mathbb{Z}\right\}. \label{eq:pL}
\end{align}
For a lattice $L\subset \mathbb{R}^n$, its \emph{dual lattice} $L^*$ is defined as 
\[
L^*\equiv \{y \in \mathbb{R}^n: \langle x,y\rangle \in \mathbb{Z},\forall x\in L\}.
\]
 For a   basis matrix ${B}=[{\vec{b}_1}\ {\vec{b}_2}\cdots {\vec{b}_n}] \in \mathbb{R}^{n\times n}$, its dual  basis matrix  is ${D}=[{\vec{d}_1}\ {\vec{d}_2}\cdots {\vec{d}_n}]= (B^T)^{-1} \in \mathbb{R}^{n\times n}$.
  Then the basis  for the dual lattice $L^*$ is $\textbf{D}= \{{\vec{d}_1},\dots,{\vec{d}_n}\}$.

For more details about lattices, please refer to \cite{regev2009notelattice}.

\subsection{Lattice Problems}
In this subsection, we introduce  two lattice problems.
In the following $\gamma=\gamma(n)\geq 1$ is called the \emph{approximation factor} of the corresponding problem.

\begin{definition}
For $\gamma=\gamma(n)\geq 1$, the search problem $\gamma$-SVP (Shortest Vector Problem)  is defined as follows: The input is a basis $\textbf{B}$ for a lattice $L \subset \mathbb{R}^n$. The goal of $\gamma$-SVP is to output a lattice point $y \in L$ such that $\|y\| \leq \gamma \cdot \lambda_1(L)$.
\end{definition}
\noindent The exact SVP is the case of $\gamma=1$. For $\gamma(n)=\left(\frac{2}{\sqrt{3}}\right)^n$, Lenstra \emph{et al.} showed that a feasible solution for $\gamma$-SVP can be found in polynomial time~\cite{lenstra1982factoring}.

\begin{theorem}[LLL algorithm]\label{LLL}
Given a basis $\textbf{B}$ for a lattice $L \subset \mathbb{R}^n$, there exists an efficient classical algorithm  $\mathsf{LLL}$ that generates a basis $\mathsf{LLL}(\bB)=\{\vec{b}'_1,\dots,\vec{b}'_n\}$ for $L$ in polynomial time such that
$$\min\limits_{i} \|{\vec{b}'_i}\| \leq \left(\frac{2}{\sqrt{3}}\right)^n\lambda_1(L).$$
\end{theorem}
\noindent 
It is obvious that the $\mathsf{LLL}$ algorithm solves the $\left(\frac{2}{\sqrt{3}}\right)^n$-SVP in polynomial time.

\begin{definition}
For $\alpha = \alpha(n) < 1/2 $, the search problem $\alpha$-BDD (Bounded Distance Decoding Problem) is defined as follows: The input is a basis $\textbf{B}$ for a lattice $L \subset \mathbb{R}^n$ and a target vector $t \in \mathbb{R}^n$ with $\dist(L,t) \leq \alpha \cdot \lambda_1(L)$. The goal of $\alpha$-BDD is to output a vector $y\in L$ such that $\|y-t\| \leq \alpha \cdot \lambda_1(L)$.
\end{definition}
\noindent  The parameter $\alpha$ is chosen so that 
$\alpha\cdot \lambda_1(L)$
is the largest decoding distance such that a lattice point (say $y$) can be recovered from a displaced vector (say $t$).

Note that most of the lattice problems become more difficult as the approximation factor $\gamma$ gets smaller, but $\alpha$-BDD becomes harder as $\alpha$ gets larger.

\subsection{Discrete Gaussian Distribution}

To solve the above mentioned lattice problems, techniques using the so-called \emph{discrete Gaussian distribution}~\cite{micciancio2007worst}, are commonly used. In the following we will introduce the  discrete Gaussian distribution and a BDD oracle built on it.

Define a function $\rho_s:\mathbb{R}^n\rightarrow \mathbb{R}$ as $$\rho_s(x) \equiv e^{\frac{-\pi \|x\|^2}{s^2}}$$
for   $s>0$.  
 When $s=1$, we will omit the subscript and simply write $\rho(x)$.
For a discrete set $A\subset \mathbb{R}^n$,  define $\rho_s(A)=\sum\limits_{a\in A} \rho_s(a)$.

\begin{definition}
   Consider a lattice $L\subset \mathbb{R}^n$ and $t \in \mathbb{R}^n$.
    The \emph{discrete Gaussian distribution} over $L + t$ with parameter $s$ is $$D_{L+t,s}(x)= \frac{\rho_s(x)}{\rho_s(L+t)}.$$
\end{definition}
\noindent (When $s=1$, we simply denote it by $D_{L+t}.$)	
Thus the probability of drawing $x\in L+t$ according to the discrete Gaussian distribution is proportional to $\rho_s(x)$. As the standard deviation $s$ grows, the discrete Gaussian distribution would become ``smoother." In \cite{micciancio2007worst}, Micciancio and Regev showed for large enough $s$, $D_{L+t,s}$ behaves in many respects like a continuous one.
To quantify how smooth the discrete Gaussian distribution is, they define a \emph{smoothing parameter} as follows.

\begin{definition}
Suppose  $L \subset \mathbb{R}^n$ is a lattice space.  For $\epsilon > 0$, the \emph{smoothing parameter} $\eta_\epsilon(L)$   is the unique value satisfying $\rho_{1/\eta_\epsilon(L)} (L^* \setminus \{ 0 \}) \leq \epsilon$.
\end{definition}
\bmark{}

\begin{definition}
Suppose $\sigma$ is a function that maps lattices to non-negative real numbers. Let  $\epsilon = \epsilon(n) \geq 0$, and $m = m(n) \in \mathbb{N}$. The problem $\epsilon$-DGS$^m_\sigma$ (Discrete Gaussian Sampling) is defined as follows: The input is a basis $\textbf{B}$ for a lattice $L \subset \mathbb{R}^n$ and a parameter $s > \sigma(L)$. The goal of $\epsilon$-DGS$^m_\sigma$ is to output   $m$ vectors so that the joint distribution for these vectors is $\epsilon$-close to $D_{L,s}$.
\end{definition}

Following the method in~\cite{2014arXiv1409.8063D} to construct a $\BDD$ solver, we first define a \emph{periodic Gaussian function}
 $$f_L(t)\equiv \frac{\rho(L+t)}{\rho(L)}$$
  for a lattice $L$.
  It is not hard to see that $f_L$ is periodic over $L$: for  $x \in L$, we have $f_L(x+t)=f_L(t).$ The idea of periodic Gaussian function was introduced by Aharonov and Regev~\cite{aharonov2005lattice} and was improved by Dadush \emph{et al.}  for solving the Closest Vector Problem (CVP) with preprocessing~\cite{2014arXiv1409.8063D}.
  CVP asks to output a closet lattice point to a given vector.
    In particular, when the target vector is close enough to the lattice, we can use a periodic Gaussian function to find its closest lattice point.
Aharonov and Regev found that the Poisson summation formula gives the identity \cite{aharonov2005lattice}:
\[
f_L(t)=\mathbb{E}_{w\sim D_{L^*}}[\cos(2\pi\langle w,t\rangle)].
\]
Hence  $f_L(t)$ can be approximated  by
\[
f_W(t)\equiv \frac{1}{N}\sum\limits_{i=1}^N \cos(2\pi\langle w_i,t\rangle),
\]
where $W=(w_1,\dots,w_N) \subset L^*$ are independent and identically-distributed (i.i.d.) samples from $D_{L^*}$ for sufficiently large $N$.
When $N$ is $O(\poly(n))$ for $n$ large enough, $f_W$ will approximate $f_L$ in statistical distance with high probability~\cite{aharonov2005lattice}.

Dadush \emph{et al.} used the above idea to construct  a $\BDD$ solver using periodic Gaussian functions \cite[Theorem 3.1]{2014arXiv1409.8063D}
 as shown in Algorithm \ref{BDD_Oracle}.
Line \ref{alg:gradientascent} is the step of \emph{gradient ascent}, which is used to approach a local maximum of a function and is explained as follows.
\begin{theorem}\label{gradientascent}\cite[Proposition 3.2]{2014arXiv1409.8063D}
Let $L \subset \mathbb{R}^n$ be a lattice with $\rho(L)=1+\epsilon$ for $\epsilon \in (0,1/200)$. Let $s_\epsilon=(\frac{1}{\pi}\log\frac{2(1+\epsilon)}{\epsilon})^{\frac{1}{2}}$, $\delta_{max}=\frac{1}{2}-\frac{2}{\pi s^2_\epsilon}$, and $\zeta(t)=\max\{1/8,\|t\|/s_\epsilon\}$. Let $W=(w_1,\dots,w_N)$ be sampled independently from $D_{L^*}$. If $N=\Omega(n\log (1/\epsilon)/\sqrt{\epsilon})$, then with probability at least $1-2^{-\Omega(n)}$,
\[
\left\|\frac{\nabla f_W(t)}{2\pi f_W(t)}+t\right\|\leq \epsilon^{(1-2\zeta(t))/4}\|t\|
\]
  for  $t \in \mathbb{R}^n$ with $\|t\|\leq \delta_{max} s_\epsilon$.
\end{theorem}
Theorem~\ref{gradientascent} shows that for any vector $t\in \mathbb{R}^n$ that is not too far from $L$, one can find a vector $t'$ closer to $L$ by doing   gradient ascent on the periodic Gaussian function. Once it become close enough to $L$,  we can find the closest vector in only polynomial time \cite{babai1986lovasz}.
In the proof of \cite[Theorem 3.1]{2014arXiv1409.8063D}, Dadush \emph{et al.} showed that for any vector $t$ is not too far from $L$, one can find the closet lattice point to $t$ by doing gradient ascent twice.
For more details, please refer to \cite{2014arXiv1409.8063D}.

\begin{algorithm}[h]

\SetAlgoLined

\SetKwData{Left}{left}\SetKwData{This}{this}\SetKwData{Up}{up}

\SetKwFunction{Union}{Union}\SetKwFunction{FindCompress}{FindCompress}

\SetKwInOut{Input}{input}\SetKwInOut{Output}{output}

\Input{lattice L(\textbf{B}), target vector t}

\Output{closest vector cv}

function Round$: \mathbb{R}^n\rightarrow \mathbb{Z}^n$ that rounds every element of an input vector;

Initialize: $count=0$;

Preprocessing: $W=(w_1,\dots,w_N)$   sampled independently from $D_{L^*}$;

\While{$count<2$}{

$f_W(t)\equiv \frac{1}{N}\sum\limits_{i=1}^N cos(2\pi\langle w_i,t\rangle)$;

$t=\frac{\nabla f_W(t)}{2\pi f_W(t)}+t$; \label{alg:gradientascent}

$count$++;

}
cv=$\mathrm{B}\cdot$ Round($\mathrm{B}^{-1}t$); // $\mathrm{B}$ is the basis matrix of $L$

\Return cv;
\caption{BDD solver constructed from a periodic Gaussian function}

\label{BDD_Oracle}

\end{algorithm}

\subsection{Quantum operators and some quantum algorithms}
In this paper we use the Dirac ket-bra notation. A qubit is a unit vector in $\mathbb{C}^2$ with two (ordered) basis vectors $\{\ket{0},\ket{1}\}$. $I=\begin{bmatrix}1 &0\\0&1\end{bmatrix}, X=\begin{bmatrix}0 &1\\1&0\end{bmatrix},  Z=\begin{bmatrix}1 &0\\0&-1\end{bmatrix},$ and  $Y=iXZ$ are the Pauli Matrices.
A universal set of gates is
 \begin{align*}
H&= \frac{1}{\sqrt{2}}\begin{bmatrix}1 &1\\1&-1\end{bmatrix},\
S=\begin{bmatrix}1 &0\\0& i\end{bmatrix},\    T=e^{i\pi/8}\begin{bmatrix}e^{-i\pi/8} &0\\0& e^{i\pi/8}\end{bmatrix},\\
CNOT&=\ket{0}\bra{0}\otimes I+ \ket{1}\bra{1}\otimes X.
 \end{align*}
We will use  a three-qubit gate, the Toffoli gate, defined by
\[
\Tof\ket{a}\ket{b}\ket{c}=
\begin{cases}
\ket{a}\ket{b}\ket{1\oplus c}, & \mbox{if $a=b=1$;}\\
\ket{a}\ket{b}\ket{c}, & \mbox{otherwise},
\end{cases}
\]
for $a,b,c\in\{0,1\}$. Toffoli gate can be efficiently decomposed into $CNOT, H, S,$ and $T$ gates~\cite{NC00} and hence it is considered as an elementary quantum gate in this paper.
In particular, Toffoli gate together with ancilla preparation are universal for classical computation:
It is easy to see that a NAND gate can be implemented by a Toffoli gate: $\Tof\ket{a}\ket{b}\ket{1}=\ket{a}\ket{b}\ket{\text{NAND}(a,b)}$,
where $\text{NAND}(a,b)=0$, if $(a,b)=(1,1)$, and $\text{NAND}(a,b)=1$, otherwise.

\begin{definition}[Search problem]
Suppose we have a set of objects named $\{1,2,\dots, N\}$, of which some are targets. Suppose  $\cO$ is an oracle that identifies the targets. 
The goal of a search problem is to find a target $i \in \{1,2,\dots, N\}$ by making queries to the oracle $\cO$.
\end{definition}

A search problem  is called a \emph{unique} search problem if there is only one target, 
 and an \emph{unknown target} search problem if at least a target exists but the number of targets is unknown.
Grover provided a quantum algorithm, that solves a unique search problem with $O(\sqrt{N})$ queries~\cite{Grover1996rk}.  When the number of targets is unknown, Brassard \emph{et al.} provided a modified Grover algorithm that solves the search problem with $O(\sqrt{N})$ queries~\cite{boyer1996tight}, which is of the same order as the query complexity of the Grover search. In general, we will simply call these algorithms by \emph{Grover search}.

\begin{theorem}\label{thm:Q_search}[Grover search]
Suppose we have an unknown target search problem with objects $\{1,2,\dots, N\}$ and a quantum oracle $\cO:\mathbb{C}^{N}\rightarrow \mathbb{C}^{N}$ such that $\cO\ket{i}=-\ket{i}$ if $i$ is a target, and  $\cO\ket{i}=\ket{i}$, otherwise.   
Then there exists a quantum algorithm that solves  the unknown target search problem  with probability at least $1/2$ and it needs  $O(\sqrt{N})$ queries to $\mathcal{O}$.
\end{theorem}

Durr and Hoyer found that the Grover search can be used to find, in  an unsorted table of $N$ values, the index that corresponds to the  minimum
with only $O(\sqrt{N})$ queries~\cite{durr1996quantum}.

\begin{theorem}\label{thm:Q_min}
Let  $T[1,2,\dots, N]$ be an unsorted table of $N$ items, each
holding a value from an ordered set. Suppose we have a quantum oracle
$\cO$
such that $\cO\ket{i}\ket{0}=\ket{i}\ket{T[i]}$.   
Then there exists a quantum algorithm that finds the index
$y$ such that $T[y]$ is the minimum  with probability at least $1/2$  and it needs  $O(\sqrt{N})$ queries to $\mathcal{O}$.
\end{theorem}

\section{Enumeration Algorithm and BDDP Oracle}
In this section we will introduce a classical algorithm $\enum$ that makes a list of the lattice points within a specified distance (called \emph{enumeration radius}) to a target vector.
Thus we can use $\enum$  to solve SVP by choosing a suitable  enumeration radius.
In  Section~\ref{sec:enum}, we will give the algorithm  $\enum$, which is built on  a classical oracle $\bdd$.
To reduce the time complexity,  in Section~\ref{sec:BDDP}, we will construct a BDDP solver,  which is a $\BDD$ solver with preprocessing, and analyze the total time complexity and space complexity.
We conclude in Section~\ref{sec:cSVP} with a  classical algorithm $\enump$ for SVP with time complexity $O(2^{2.0478n})$ and space complexity $O(2^{n/2})$. To the best of our knowledge, the space complexity of our algorithm $\enump$ is more efficient than other single exponential time classical algorithms. (See Table~\ref{tb:SVP_result} for a comparison.)

\subsection{Enumeration Algorithm $\enum$} \label{sec:enum}

Given a lattice $L\subset \mathbb{R}^n$, a target vector $t\in \mathbb{R}^n$, and a parameter $\delta$, we would like to find all the lattice points within  distance $\delta$ to $t$.
If this can be done,  SVP can be reduced to an enumeration problem:
make a list of lattice points that are within  distance $\delta>\lambda_1(L)$ to the origin and then find the shortest vector  in the list.
It may be difficult to generate the list at a first glance, since there are countably infinite lattice points.
However, inspired by Kirchner and Fouque's lattice enumeration algorithm~\cite{kirchner2016time}, we find that it suffices to consider only a finite number of lattice points by using the properties of BDD oracle.

\begin{lemma} \label{lemma:y_equiv}
Consider a   lattice  $L \subset \mathbb{R}^n$, $0.5 > \alpha=\alpha(L)>0$,  and $t \in \mathbb{R}^n$, satisfying $\dist(L,t)<\alpha  \lambda_1(L)$. Suppose $\bdda$ is an $\alpha$-$\BDD$ oracle.
 Let  $p$ be an integer  such that $p\alpha\geq 1$.
Then we have
  \begin{align}
  y = (y \mod pL) - p\cdot\bdda\left(L, \left(\frac{y}{p} \mod L\right) - \frac{t}{p}\right)\label{eq:y_equiv}
  \end{align}
  for any  $y \in \mathbb{R}^n$ such that $\|y- t\|<p\alpha  \lambda_1(L)$, where $pL$ is defined in (\ref{eq:pL}).
\end{lemma}
\begin{proof}
{
An $\alpha$-$\BDD$ oracle $\bdda$ with $\alpha <0.5$ will have the following properties:
 \begin{align}
 \bdda(L,t+x)=&\bdda(L,t)+x \label{eq:shift}\\
 \bdda(pL, p t) =& p\cdot\bdda(L, t), \label{eq:scaling}
 \end{align}
 for any $x \in L$, and $t \in \mathbb{R}^n$ such that $\dist(L,t)<\alpha  \lambda_1(L)$.
First we prove (\ref{eq:shift}). For   $t \in \mathbb{R}^n$ satisfying $\dist(L,t)<\alpha  \lambda_1(L)<0.5 \lambda_1(L)$, $\bdda(L,t)$ returns a unique valid lattice point $v$ that satisfies $\|v-t\|=\dist(L,t)$. By definition, we know  $\dist(L,t)=\min\limits_{x\in L}\|x-t\|=\min\limits_{x'\in L}\|x'-t+a\|=\dist(L,t+a)<\alpha  \lambda_1(L)<0.5\lambda_1(L)$ for any $a \in L$. Then we know $\bdda(L,t+a)$ returns a unique valid lattice point $v'$, which satisfies $\|v'-t+a\|=\dist(L,t+a)=\dist(L,t)=\|v-t\|$,
 which implies Equation (\ref{eq:shift}).

As for (\ref{eq:scaling}), suppose $\dist(pL,pt)\leq  \alpha  \lambda_1(pL) < 0.5p\lambda_1(L) $.  Then $\bdda(pL,pt)$ returns a unique valid lattice point $w$. Also $\dist(pL,pt)=\min\limits_{x\in pL}\|x-pt\|=p\cdot\min\limits_{x'\in L}\|x'-t\|=p\cdot \dist(L,t)\leq p\alpha  \lambda_1(L)<0.5p\lambda_1(L) $. Therefore $\bdda( L,t)$ returns a unique valid output $w'\in L$ satisfying $p\cdot\dist(L,t)=\|pw'-pt\|=\dist(pL,pt)$, which implies Equation (\ref{eq:scaling}).

Now we prove  Equation~(\ref{eq:y_equiv}). Suppose $y \in \mathbb{R}^n$ and $\|y-t\|<p\alpha  \lambda_1(L)<0.5p\lambda_1(L) $.
By definition,
\begin{eqnarray*}
0 & = & \bdda(pL, y-t) \\
 & = & \bdda\left(pL, (y - (y \mod pL)) + (y \mod pL) - t \right)  \\
 &  \stackrel{(a)}{=} & \bdda(pL, (y \mod pL) - t) + y - (y \mod pL) \\
 &  \stackrel{(b)}{=} & p\cdot\bdda(L, (\frac{y}{p} \mod L) - \frac{t}{p}) + y - (y \mod pL),
 \end{eqnarray*}
 where $(a)$ is by (\ref{eq:shift}) and $(b)$ is by (\ref{eq:scaling}).

 Therefore, $y = (y \mod pL) - p\cdot\bdda(L, (\frac{y}{p} \mod L) - \frac{t}{p})$}.
\end{proof}

  {\bf{Remark:}} According to the above lemma, we can enumerate the lattice points within   distance $p\alpha  \lambda_1(L)$ to the origin by
checking the lattice points of the coset leaders in $L/pL$ with the help of an $\alpha$-BDD oracle, where $p\alpha\geq1$.
More precisely, if $y'=y+b$ for $b\in pL$, then $\bdda(pL,y'-t)=\bdda(pL,y-t)+b$ by (\ref{eq:shift}).
Also $\left|L/pL\right|=p^n$, so at most $p^n$ queries to the $\alpha$-BDD oracle are needed.

Therefore, we have Algorithm $\enum$   defined in Algorithm~\ref{EnumAlgorithm}.
\begin{theorem}[Algorithm \enum]\label{Enum}
Given a lattice $L \subset \mathbb{R}^n$ with basis matrix $\mathrm{B}$, a target vector $t\in\mathbb{R}^n$, an $\alpha$-$\BDD$ oracle $\bdda$ with $\alpha < 0.5$, and an integer scalar $p$ such that $p\alpha>1$,  $\enum$ defined in Algorithm \ref{EnumAlgorithm} collects all lattice points within distance $p\alpha \lambda_1(L)$ to $t$ (and some other lattice points).
\end{theorem}

\begin{proof}
{
It suffices to show that for any lattice point $y \in L$ satisfying $\|y-t\|<p\alpha\lambda_1(L)$, there exists $s\in\mathbb{Z}_p^n$ such that $y=-p\cdot\bdda(L,\frac{\mathrm{B}s-t}{p})+\mathrm{B}s$.
Suppose  $y \in L$ such that $\|y-t\|<p\alpha\lambda_1(L)$. By Lemma~\ref{lemma:y_equiv}, we have
\[
y = (y \mod pL) - p\cdot\bdda(L, (\frac{y}{p} \mod L) - \frac{t}{p}).
\]
Since  $(y \mod pL)$ can be represented by $\mathrm{B}s$ for some $s\in \mathbb{Z}^n_p$, we can rewrite the above equation  as
\[
y = -p\cdot\bdda(L,\frac{\mathrm{B}s-t}{p})+\mathrm{B}s,
\]
and there we can using points in $L/pL$ to compute all lattice points close enough to $t$}
\end{proof}

Note that the list generated in Theorem~\ref{Enum} may contain lattice points whose distance to $t$ is greater than $p\alpha \lambda_1(L)$.

\begin{algorithm}[h]

\SetAlgoLined

\SetKwData{Left}{left}\SetKwData{This}{this}\SetKwData{Up}{up}

\SetKwFunction{Union}{Union}\SetKwFunction{FindCompress}{FindCompress}

\SetKwInOut{Input}{input}\SetKwInOut{Output}{output}

\Input{lattice $L$, basis matrix $\mathrm{B}$, target vector $t$, scalar $p$, BDD oracle $\bdda$}

\Output{A set of vectors $OUTPUT$ that contains the lattice points within  distance $p\alpha\lambda_1(L)$ to  $t$}

Initialize: OUTPUT=null array\;

\For{all $s \in Z^n_p$}{

OUTPUT[$s$]=$-p\cdot\bdda(L,\frac{\mathrm{B}s-t}{p})+\mathrm{B}s$;

}
\Return OUTPUT;
\caption{The enumeration algorithm   $\enum(L(\textbf{B})$, $t$, $p$, $\bdda$)}

\label{EnumAlgorithm}

\end{algorithm}

\subsection{Constructing a 0.391-BDD solver} \label{sec:BDDP}

In this subsection, we will construct an $\alpha$-BDD solver with preprocessing for some $1/3\leq \alpha< 1/2$.
Recall that the algorithm $\enum$ makes $\lceil\frac{1}{\alpha}\rceil^n$ queries to an $\alpha$-$\BDD$ oracle $\bdda$
if $p$ is chosen to be $\lceil\frac{1}{\alpha}\rceil$.
Then  using the algorithm $\enum$ to solve SVP will take time $O(B_\alpha\cdot \lceil\frac{1}{\alpha}\rceil^n)$,
where   $B_\alpha$ is the running time of $\bdda$.

Now we determine how small alpha can be. Aggarwal \emph{et al.} have the following reduction from $\alpha$-$\BDD$ to DGS:

\begin{theorem}\label{CVPtoDGS}\cite[Theorem 7.3]{DBLP_journals_corr_AggarwalDS15}
For any $\epsilon \in (0,1/200)$, let
\[
\alpha(L)\equiv \frac{\sqrt{\log(1/\epsilon)/\pi-o(1)}}{2\eta_\epsilon(L^*)\lambda_1(L)}.
\]
Then there exists a reduction from $\alpha$-$\BDD$ to $\frac{1}{2}$-DGS$^m_{\eta_\epsilon(L^*)}$, where $m=O\left(n\log(1/\epsilon)/\sqrt{\epsilon}\right)$. The reduction, which preserves the dimension, makes a single query to the DGS oracle, and runs in time $m\cdot \text{poly}(n)$.
\end{theorem}

By Theorem~\ref{CVPtoDGS}, we know the quality of DGS will determine how many arithmetic operations we need to compute the periodic Gaussian function and   how large $\alpha(L)$ we have.  
Aggarwal \emph{et al.} proposed a construction of DGS in  \cite{DBLP_journals_corr_AggarwalDS15}.



\begin{theorem}\label{hDGS}\cite[Theorem 5.11]{DBLP_journals_corr_AggarwalDS15}
For a lattice $L \subset \mathbb{R}^n$, let $\sigma(L)=\sqrt{2}\cdot \eta_{1/2}(L)$. Then there exists an algorithm that solves $\exp(-\Omega(\kappa))$-DGS$^{2^{n/2}}_\sigma$ in time $2^{\frac{n}{2}+\text{polylog}(\kappa)+o(n)}$ with space $O(2^{n/2})$ for any $\kappa\geq \Omega(n)$.
\end{theorem}

Basically, their idea is to sample $O(2^{n/2})$ vectors from $D_{L,\sigma}$.
We will use their DGS solver with  
 $s=\sigma+o(1)$ in our algorithm and the above theorem states that we can prepare $O(2^{n/2})$ vectors from $D_{L,\sigma}$ in time $O(2^{n/2})$ when $s>\sqrt{2}\cdot \eta_{1/2}(L)$.

Suppose we have $O(2^{n/2})$ discrete Gaussian samples with   standard deviation greater than or equal to $\sqrt{2}\cdot \eta_{1/2}(L^*)$ for any $L^* \subset \mathbb{R}^n$.
  Now we want to construct a $(1/3)$-$\BDD$ solver (or above) to solve SVP by using Theorem \ref{CVPtoDGS},
  so we need to give a lower bound for $\alpha(L)$ of a BDD oracle on condition that the standard deviation of discrete Gaussian samples is $\sqrt{2}\cdot \eta_{1/2}(L)$.
   The following lemma   provides a relation between the smoothing parameter $\eta_\epsilon(L^*)$ and  $\lambda_1(L)$:

\begin{lemma}\label{smoothinglambda}\cite[Lemma 6.1]{DBLP_journals_corr_AggarwalDS15}
For any lattice $L \subset \mathbb{R}^n$ and $\epsilon \in (0,1)$, if $\epsilon > (e/\beta^2+o(1))^{-\frac{n}{2}}$, where $\beta= 2^{0.401}$, we have
 \begin{align}
\sqrt{\frac{\log(1/\epsilon)}{\pi}} < \lambda_1(L) \eta_\epsilon(L^*) < \sqrt{\frac{\beta^2n}{2\pi e}}\cdot \epsilon^{-1/n}\cdot(1+o(1)), \label{ineq:smoothingtosvlarge}
 \end{align}
and if $\epsilon \leq (e/\beta^2+o(1))^{-\frac{n}{2}}$, we have
 \begin{align}
\sqrt{\frac{\log(1/\epsilon)}{\pi}} < \lambda_1(L) \eta_\epsilon(L^*) < \sqrt{\frac{\log(1/\epsilon)+n\log\beta+o(n)}{\pi}}. \label{ineq:smoothingtosvsmall}
 \end{align}

\end{lemma}
 We know that we can have an $\alpha$-BDD oracle from Algorithm~\ref{BDD_Oracle}
 with $\alpha$ as large as $0.5-o(1)$.
This is because the largest decoding distance is, by Theorem~\ref{CVPtoDGS},  $(\sqrt{\log(1/\epsilon)/\pi-o(1)})/2\eta_\epsilon(L^*),$
which is less than $\lambda_1/2$ by Lemma~\ref{smoothinglambda}.
Hence $p\geq 3$ and at least $3^n$  queries to the $\BDD$ oracle are necessary to solve SVP.

Furthermore, the smaller $\alpha$ of BDD is, the fewer arithmetic operations we need. Therefore we want to make $\alpha$ close to $1/3$.
We will use Lemma~\ref{smoothinglambda} to give a lower bound for the parameter $\alpha$ in Theorem~\ref{CVPtoDGS}.
The other parameter of our concern is $\eta_{\epsilon}(L^*)$ because in Theorem~\ref{CVPtoDGS} one can solve $\alpha(L)$-BDD with  preprocessed discrete Gaussian samples from $D_{L^*,\eta_{\epsilon}(L^*)}$.
To apply Theorem~\ref{hDGS} for DGS, we choose $s=\eta_{\epsilon}(L^*)$, which satisfies $\eta_{\epsilon}(L^*)>\sqrt{2} \eta_{\frac{1}{2}}(L^*)$.
Then we obtain the the following corollary.

\begin{corollary}\label{smoothingDGS}
Let $\epsilon= e^{-(\beta^2/e+o(1))n}$, where $\beta= 2^{0.401}$. Then there exists an algorithm that solves $exp(-\Omega(\kappa))$-DGS$^{2^\frac{n}{2}}_{\eta_\epsilon(L^*)}$ in time $2^{\frac{n}{2}+polylog(\kappa)+o(n)}$ for any $\kappa\geq \Omega(n)$.
\end{corollary}

\begin{proof}
{
Let $\epsilon'= e^{-(\beta^2/e+o(1))n}$. We now try to prove $\eta_{\epsilon'}(L^*)>\sqrt{2} \eta_{\frac{1}{2}}(L^*)$.
First when $\epsilon=1/2>(e/\beta^2+o(1))^{-\frac{n}{2}}$, by the right inequality of (\ref{ineq:smoothingtosvlarge}) we have:
\[
\sqrt{2}\eta_{1/2}(L^*)<\frac{1}{\lambda_1(L)} \cdot \sqrt{\frac{\beta^2n}{\pi e}}\cdot \epsilon^{-1/n}\cdot(1+o(1)).
\]
As for $\epsilon= \epsilon'<(e/\beta^2+o(1))^{-\frac{n}{2}}$, we use the left inequality of (\ref{ineq:smoothingtosvsmall})  to have:
\[
\eta_\epsilon'(L^*) > \sqrt{\frac{\log(1/\epsilon')}{\pi}} \cdot \frac{1}{\lambda_1(L)} =\sqrt{\frac{(\beta^2/e+o(1))n}{\pi}}\cdot \frac{1}{\lambda_1(L)}.
\]
Hence we know $\eta_\epsilon(L^*) > \sqrt{\frac{(\beta^2/e+o(1)))n}{\pi}}\cdot \frac{1}{\lambda_1(L)} > \sqrt{2}\eta_{1/2}(L^*)$. Then by Theorem $\ref{hDGS}$ we complete the proof.}
\end{proof}
Therefore, combining corollary~\ref{smoothingDGS} and theorem~\ref{CVPtoDGS}, we derive the following result.

\begin{corollary}\label{391bdd}
Let $\beta = 2^{0.401}$ and $\epsilon= e^{-(\beta^2/e+o(1))n}$. There exists an algorithm that solves $\alpha$-BDD in time $O(e^{(\beta^2/2e+o(1))n}+2^{0.5n+o(n)})=O(2^{0.5n+o(n)})$ and in space $O(2^{0.5n})$ for $\alpha = 0.391$.
\end{corollary}
\begin{proof}
{
Let $\alpha(L)\equiv \frac{\sqrt{\log(1/\epsilon)/\pi-o(1)}}{2\eta_\epsilon(L^*)\lambda_1(L)}$ and $\epsilon= e^{-(\beta^2/e+o(1))n}$.
By Lemma~$\ref{smoothinglambda}$, for any $\epsilon \leq (e/\beta^2+o(1))^{-\frac{n}{2}}$ we have:
\[
 \frac{1}{\eta_\epsilon(L^*)} > \sqrt{\frac{\pi}{\log(1/\epsilon)+n\log\beta+o(n)}} \cdot \lambda_1(L),
\]
put it in $\alpha(L)$ we have
\[
\alpha(L) > \frac{1}{2}  \sqrt{\frac{\log(1/\epsilon)-o(1)}{\log(1/\epsilon)+n\log\beta+o(n)}}.
\]
Hence by Theorem~\ref{CVPtoDGS}, once we have $O(e^{(\beta^2/2e+o(1)))n})$ samples from $D_{L^*,\eta_\epsilon(L^*)}$,   then we solves $\alpha$-BDD in time $O(e^{(\beta^2/2e+o(1))n})$ for
\[
\alpha\equiv \frac{1}{2}  \sqrt{\frac{\log(1/\epsilon)}{\log(1/\epsilon)+0.401n+o(n)}}> 0.391-o(1),
\]
which proves the corollary}

\end{proof}

Though we need a $1/3$-$\BDD$ oracle to  solve SVP by Algorithm~\ref{EnumAlgorithm} with $3^n$ queries, however when we choose  the corresponding $\epsilon$ to have a $1/3$-$\BDD$ solver, the smoothing parameter $\eta_{\epsilon(L^*)}$ cannot be proven that it is greater than $\sqrt{2}\cdot\eta_{1/2}$.
 Hence  we choose the smoothing parameter $\eta_\epsilon(L^*)$ in Theorem~\ref{CVPtoDGS} with $\epsilon$ that is derived in Corollary \ref{smoothingDGS}, and then we have a $0.391$-BDD oracle, which can be built in time $O(e^{(\beta^2/2e+o(1))n}+2^{n/2})=O(2^{0.5n+o(n)})$. Note that $0.391>1/3$ fits in our Algorithm \ref{EnumAlgorithm} for finding the shortest vector in a lattice $L$.

\subsection{Classical SVP Algorithm $\enump$}\label{sec:cSVP}

In this subsection, we will explicitly show how to solve SVP by our algorithm $\enump$.
In the last subsection, we show that one can build a $0.391$-BDD oracle in time $O(2^{0.5n})$. 
In fact, the discrete Gaussian samples can be reused for another execution of the BDD oracle.
Therefore,  we can reduce the time complexity by preparing these discrete Gaussian samples in advance.
We define BDD with preprocessing as follows:

\begin{definition}\label{BDDP}
For $\alpha = \alpha(n) < 1/2$, the search problem with preprocessing $\alpha$-BDDP is defined as follows: The problem contains two phase, preprocessing phase and query phase. The input to the preprocessing phase is a basis $\bf{B}$ for a lattice $L \subset \mathbb{R}^n$, and the output to the preprocessing phase is an advice string $A$. In the query phase the inputs are a vector $t \in \mathbb{R}^n$ and the advise string $A$ from the preprocessing process. Then $\alpha$-BDDP is the problem of solving $\alpha$-BDD with preprocessing. Only  on the running time in the query phase matters and the preprocessing phase may take arbitrary time.
\end{definition}

Consider Algorithm~\ref{BDD_Oracle}.  
In the preprocessing phase, several samples   are prepared according to a specific discrete Gaussian distribution on the dual space; in the query phase $f_W(t)\equiv \frac{1}{N}\sum\limits_{i=1}^N cos(2\pi\langle w_i,t\rangle)$ is used to  approximate the periodic Gaussian function,
and gradient ascent is conducted on the periodic Gaussian to construct a BDD solver.  
Therefore, the 0.391-BDDP solver takes time   $O(2^{0.4628n})$  when DGS is prepared in advance,
while  a full 0.391-BDD solver takes time $2^{0.5n+o(n)}$.
Hence we  have 
the following corollary.

\begin{corollary}\label{cor:BDDP}
Let $\beta = 2^{0.401}$ and $\epsilon= e^{-(\beta^2/e+o(1))n}$. There exists an algorithm that solves $\alpha$-BDDP in time $O(e^{(\beta^2/2e+o(1))n})$ for $\alpha = 0.391$. The preprocessing algorithm generates  $O(e^{(\beta^2/2e))n})$ samples from $D_{L^*,\eta_\epsilon(L^*)}$.
\end{corollary}

\begin{algorithm}[h]

\SetAlgoLined

\SetKwData{Left}{left}\SetKwData{This}{this}\SetKwData{Up}{up}

\SetKwFunction{Union}{Union}\SetKwFunction{FindCompress}{FindCompress}

\SetKwInOut{Input}{input}\SetKwInOut{Output}{output}

\Input{lattice L(\textbf{B})}

\Output{shortest vector $sv$}

Initialize: $\epsilon= 2^{-(\beta^2/e+o(1))n}$,\   sv=$\inf$\;

Preprocessing: $O(2^{n/2})$ discrete Gaussian samplings from $D_{L^*,\eta_\epsilon}$; // Corollary \ref{smoothingDGS}.

\For{all $s \in \mathbb{Z}^n_3$}{

\If{$sv$ $>$ $\bddp_{0.391}(3L,\mathrm{B}s)+\mathrm{B}s \hspace{2mm}\mbox{and}\hspace{2mm} \bddp_{0.391}(3L,\mathrm{B}s)+\mathrm{B}s \neq 0$
}{

$sv = \bddp_{0.391}(3L,\mathrm{B}s)+\mathrm{B}s$; //$\mathrm{B}$ is the basis matrix of $L$

//A BDDP algorithm needs $O(e^{(\beta^2/2e+o(1))n})$ arithmetic operations by Corollary \ref{cor:BDDP}.

}

}
\Return $sv$;
\caption{The algorithm $\enump$ that solves SVP}

\label{C_SVPalgorithm}
\end{algorithm}

Combining Theorem~\ref{Enum} and Corollary~\ref{cor:BDDP}, we then have a classical algorithm $\enump$ in Algorithm~\ref{C_SVPalgorithm} for
SVP.

\begin{theorem}\label{C_SVPsolver}
There exists a classical probabilistic algorithm that solves SVP with probability $1-2^{-\Omega(n)}$ in time $O(e^{(\beta^2/2e+o(1))n}\cdot 3^n)=2^{2.0478n+o(n)}$ with space $2^{n/2+o(n)}$, where $\beta=2^{0.401}$.
\end{theorem}

\begin{proof}
{
Consider $\enump$. For a lattice $L \subset \mathbb{R}^n$, Corollary~\ref{cor:BDDP} provides a $0.391$-BDDP algorithm in time $O(e^{(\beta^2/2e+o(1))n})$. By Theorem $\ref{Enum}$,   we can enumerate all the lattice points with length less than $1.173 \lambda_1(L)$ by using the $0.391$-BDDP algorithm $3^n$ times. Also by Corollary $\ref{smoothingDGS}$ we   need $O(2^{n/2})$ time to prepare $\exp(-\text{poly}(n))$-DGS$^{O(2^\frac{n}{2})}_{\eta_\epsilon(L^*)}$. Thus $\enump$ reports a shortest vector in time $O(2^{n/2}+3^n\cdot e^{(\beta^2/2e+o(1))n}=2^{2.0478n+o(n)}$ with space $2^{n/2+o(n)}$}

\end{proof}

\section{Quantum speedup for enumeration algorithm}

In the previous section, we provided a classical algorithm $\enump$ that solves SVP in time $2^{2.0478n+o(n)}$.
In $\enump$, the lattice points of length less than $1.173 \lambda_1(L)$  are collected and compared so that the shortest vector is recorded.
We will adapt the classical algorithm $\enump$ to a quantum version, using a modified Grover search (Theorem~\ref{thm:Q_search}) that can find a
 nonzero vector with the shortest length with high probability.
We will first explain the main idea of our quantum algorithm $\qsvp$ for solving SVP in Section~\ref{sec:4.1}. 
Then  
we will  introduce the quantum enumeration algorithm by constructing a quantum circuit $\cO_d$
that identifies a lattice point with length less than $d\lambda_1$.
$\cO_d$ can be used to solve $d$-SVP with constant successful probability by using $2^{1.2553n+o(n)}$ Toffoli gates and classical space $2^{n/2+o(n)}$.

As a consequence, we can use the idea of minimum finding algorithm (Theorem~\ref{thm:Q_min}) to solve exact SVP by recursively using $\cO_d$.
A key component of $\cO_d$ is a \emph{filter circuit}  $\Filter_d$, which will be explicitly constructed in Section~\ref{sec:gap_SVP}. The filter circuit identifies the vectors with length in a specific range. At the last step we will show how to recursively update $\Filter_d$ with a smaller $d$ to implement the desired oracle operation $\cO$ and consequently we can solve  exact SVP. 

\subsection{Main idea of $\qsvp$} \label{sec:4.1}

Suppose $L\subset \mathbb{R}^n$ is a lattice space with basis matrix $B$.
Define a function $f_{BDDP(L)}:\mathbb{Z}^n_3\rightarrow \mathbb{R}^n$ as:
\begin{align}
f_{BDDP(L)}(s) \equiv -\bddp_{0.391}(3L,Bs)+Bs  \label{eq:fbddl}
\end{align}
for $s\in \mathbb{Z}^n_3$, where $\bddp_{0.391}$ is given in Corollary~\ref{cor:BDDP}.
Observe that in $\enump$,  $3^n$ queries are made to   $\bddp_{0.391}$ to compute $f_{BDDP(L)}(s)$ for all $s\in \mathbb{Z}^n_3$. Suppose we have a quantum circuit $U_{BDDP(L)}$ that computes $f_{BDDP(L)}$ defined by 
\begin{align}U_{BDDP(L)}\ket{i}\ket{x}=\ket{i}\ket{x\oplus \| f_{BDDP(L)}(i)\|},
\end{align}
for $i\in \mathbb{Z}^n_3$ and $x \in \mathbb{R}^n$,
where the second register has $b$ qubits to represent $f_{BDDP(L)}(s)$ over $s \in \mathbb{Z}^n_3$.
Then we can prepare the superposition state $$\ket{\psi}=\frac{1}{3^{n/2}}\sum_{i \in \mathbb{Z}^n_3}\ket{i}\ket{0}$$ and run the circuit $U_{BDDP(L)}$ once to obtain
\begin{align}
U_{BDDP(L)}\ket{\psi}=&\frac{1}{3^{n/2}}\sum_{i \in \mathbb{Z}^n_3}\ket{i}\ket{\|f_{BDDP(L)}(i)\|}. \label{eq:UBDDP}
\end{align}
However it is difficult to directly find one $\ket{i}$ such that its length  (represented by $\ket{\|f_{BDDP(L)}(i)\|}$)  is equal to $ \lambda_1(L)$.
This is how the quantum Grover search enters our discussion.
If we can efficiently apply the quantum Grover search to identity a target vector, then we have a desired quantum SVP solver.
Recall that the Grover search has two rotation operators: one is the oracle operation $\cO$ that reflects a state about the indices;
the other is a rotation
\begin{align}\cG= H^{\otimes m} \left(2 \ket{0}^{\otimes m}\bra{0}^{\otimes m} - \mathbb{I}_m\right)  H^{\otimes m}
\end{align}  about the superposition state of all solutions to the search problem,
where $m$ is the number of qubits of the index space and $\mathbb{I}_m$ is the identity operator on the $m$-qubit space.
Then it remains to construct an quantum algorithm for the oracle operation $\cO$.

\subsection{Quantum enumeration algorithm}\label{sec:QEnum}

To construct $\cO$, the first step is to identify whether the length of a vector is $\lambda_1(L)$ or not.
However,  this is difficult since we do not know the value of $\lambda_1(L)$ and to determine this value is NP hard \cite{Haviv:2007:THS:1250790.1250859}.
Alternatively, we construct a quantum circuit $\cO_d:\mathbb{C}^{2^{m+a'}}\rightarrow\mathbb{C}^{2^{m+a'}}$:
\begin{align}
\cO_d \ket{i}\ket{0}^{\otimes a'}=
\begin{cases}
-\ket{i}\ket{0}^{\otimes a'}, & \mbox{if $\lambda_1(L)\leq \| f_{BDDP(L)}(i)\|<d \cdot \lambda_1(L)$;} \\
\ket{i}\ket{0}^{\otimes a'}, & \mbox{otherwise,}
\end{cases}
\end{align}
for $i \in \mathbb{Z}^n_3$,
where $d>1$, $m\equiv n\cdot\log_23 $, and $a'\in \mathbb{N}$ is the number of ancilla qubits.
Suppose we have access to the circuit $\cO_d$. Then  by Theorem \ref{thm:Q_search},    we can solve $d$-SVP with $O(3^{n/2})$ uses of $\cO_d$. 
The resulting algorithm $\qenump$ is illustrated in Fig.~\ref{fig:UnknownGroverSearch}. (Details of $\qenump$ are postponed to Algorithm~\ref{QdSVPsolver} after we give the construction of $\cO_d$.)
 The output of  $\qenump$ is an index $i$ with $\lambda_1\leq \|f_{BDDP_L}(i)\|<d\lambda_1$. 
We would like to have $\qenump$  with an output index that  corresponds exactly to a vector with length $\lambda_1$.
Our method is to recursively use the quantum circuit $\cO_d$ with $d$ updated adaptively to achieve our goal.
It basically follows Theorem~\ref{thm:Q_min} to find a nonzero minimum over $3^n$ indices.
We will   get an index $i$ with $\|f_{BDDP_L}(i)\|=\lambda_1$ with very high probability and use it to build the Grover search oracle $\cO$.


\begin{figure*}[h]
\[
\Qcircuit @C=1em @R=.7em {
  \lstick{\sum_i\ket{i}} & \qw  & \multigate{1}{\cO_d} & \gate{\cG}    & \qw  & \cdots & & \meter & \cw \\
  \lstick{\ket{0}^{\otimes a}}   & \qw  & \ghost{\cO_d}  & \qw   & \qw & \cdots & & \rstick{\ket{0}^{\otimes a}} \qw\\
             &      &        &      & \dstick{\text{Repeat $3^{n/2},3^{n/2}/\sqrt{2},\cdots$ times}}
  \gategroup{1}{3}{2}{7}{.7em}{--}
 }\]
\vspace{0.5cm}
\fcaption{Quantum algorithm that solves $d$-SVP, \qenump}
\label{fig:UnknownGroverSearch}
\end{figure*}

Now we get into the details.
 Our quantum circuit $\cO_d$ has two components. The first one is the quantum circuit $U_{BDDP(L)}$ defined in (\ref{eq:UBDDP}).
 The quantum complexity of $U_{BDDP(L)}$ circuit is same as the amount of arithmetic operations we need for computing classical $\BDDP$ oracle. By Corollary~\ref{cor:BDDP}, we can use $2^{0.4629n+o(n)}$ arithmetic operations to compute $\BDDP_{0.391}$, where $\beta = 2^{0.401}$. Hence we can use $2^{0.4629n+o(n)}$ Toffoli gates to execute the quantum circuit $\BDDP$ with preprocessed $2^{0.4629n+o(n)}$ vectors that are sampled from $D_{L^*,\eta_\epsilon(L^*)}$, where $\epsilon= 2^{-(\beta^2/e+o(1))n}$.

 The second component is a filter circuit
  defined as:
\begin{align}
\Filter_d \ket{v}\ket{0}^{\otimes a}=
\begin{cases}
-\ket{v}\ket{0}^{\otimes a}, & \mbox{if $\lambda_1(L)\leq  v < d \cdot \lambda_1(L)$;} \\
\ket{v}\ket{0}^{\otimes a}, & \mbox{otherwise,}
\end{cases}
\end{align}
for a lattice $L \subset \mathbb{R}^n$,
where $d  >1$, $v \in \mathbb{R}^n$,  $\ket{v}$ is in a representation of appropriate dimension, and $a$ is the number of ancilla qubits.
The circuit $\Filter_d$   can filter out the candidates for solving $d$-SVP.
(Thus we would like to have  $\Filter_d$  with $d$ close to $1$ as possible.) 
 Details of $\Filter_d$ are postponed to Section~\ref{sec:gap_SVP}.

Then $\cO_d$ is defined by
$$\cO_d\equiv (U_{BDDP(L)}\otimes \mathbb{I}_{a}) \cdot (\mathbb{I}_m \otimes \Filter_d) \cdot (U_{BDDP(L)} \otimes \mathbb{I}_{a})$$ for a lattice $L \subset  \mathbb{R}^n$.
Figure~\ref{fig:dSVP} illustrates the components of $\cO_d$. Let $len_i= \| f_{BDDP(L)}(i)\|$ here.
It can be checked that 
for $i \in \mathbb{Z}^n_3$, we have:
\begin{align*}
\cO_d\ket{i}\ket{0}^{\otimes b}\ket{0}^{\otimes a}
&= (U_{BDDP(L)}\otimes \mathbb{I}_a) \cdot (\mathbb{I}_m \otimes \Filter_d) \cdot (U_{BDDP(L)} \otimes \mathbb{I}_a)\ket{i}\ket{0}^{\otimes b}\ket{0}^{\otimes a} \\
&=  (U_{BDDP(L)}\otimes \mathbb{I}_a) \cdot (\mathbb{I}_m \otimes \Filter_d) \ket{i}\ket{ len_i}\ket{0}^{\otimes a}\\
&= \begin{cases}
 (U_{BDDP(L)}\otimes \mathbb{I}_a) \cdot(-1)\cdot\ket{i}\ket{len_i}\ket{0}^{\otimes a} , & \mbox{if $\lambda_1(L)\leq len_i<d  \lambda_1(L)$;} \\
 (U_{BDDP(L)}\otimes \mathbb{I}_a) \ket{i}\ket{len_i}\ket{0}^{\otimes a}, & \mbox{otherwise,}
\end{cases} \\
&=\begin{cases}
 -\ket{i}\ket{len_i\oplus len_i}\ket{0}^{\otimes a}, & \mbox{if $\lambda_1(L)\leq len_i<d\lambda_1(L)$;}\\
  \ket{i}\ket{len_i\oplus len_i}\ket{0}^{\otimes a}, & \mbox{otherwise,}
\end{cases}\\
&=  \begin{cases}
 -\ket{i}\ket{0}^{\otimes b}\ket{0}^{\otimes a}, & \mbox{if $\lambda_1(L)\leq len_i<d  \lambda_1(L)$;}\\
  \ket{i}\ket{0}^{\otimes b}\ket{0}^{\otimes a}, & \mbox{otherwise,}
\end{cases}
\end{align*}
as desired.

\begin{figure*}[h]

\centering
\[
\Qcircuit @C=1em @R=.7em {
	\lstick{\ket{i}}       & \qw & \multigate{1}{U_{BDDP(L)}}  &  \qw  &  \qw                    &  \qw  & \multigate{1}{U_{BDDP(L)}}  &   \rstick{\pm\ket{i}}      \qw \\
	\lstick{\ket{0}^{\otimes b}} & \qw & \ghost{U_{BDDP(L)}}         &  \qw  &  \multigate{1}{\Filter_d}  &  \qw  & \ghost{U_{BDDP(L)}}         &   \rstick{\ket{0}^{\otimes b}} \qw \\
	\lstick{\ket{0}^{\otimes a}}       & \qw & \qw                 &  \qw  &  \ghost{\Filter_d}         &  \qw  & \qw                 &   \rstick{\ket{0}^{\otimes a}}       \qw
}\]
\vspace{0.5cm}
\fcaption{The quantum circuit  $\cO_d$.}
\label{fig:dSVP}
\end{figure*}

\begin{algorithm}

\SetAlgoLined

\SetKwData{Left}{left}\SetKwData{This}{this}\SetKwData{Up}{up}

\SetKwFunction{Union}{Union}\SetKwFunction{FindCompress}{FindCompress}

\SetKwInOut{Input}{input}\SetKwInOut{Output}{output}

\Input{lattice $L(\bB)$, $d'$, confidence parameter $\kappa$}

\Output{short vector length $sv$}

Initialize: quantum gate $U_{BDDP(L)}$;

Initialize: quantum gate $\Filter_{d}$ with filter number $d=d'/\lambda_1(L)$;

Initialize: quantum gate $\cO_{d}=(U_{BDDP(L)}\otimes \mathbb{I}_a) \cdot (\mathbb{I}_m \otimes \Filter_{d}) \cdot (U_{BDDP(L)} \otimes \mathbb{I}_a)$;

Initialize: $candidate.$index =$\emptyset$, $candidate$.length $=\inf$;

\Repeat{$\kappa$ times}{

    T=$3^{n/2}$;

Initialize: $\ket{\psi}=\frac{1}{3^{n/2}}\sum_{i \in \mathbb{Z}^n_3}\ket{i}\ket{0}^{\otimes b}\ket{0}^{\otimes a}$;

	\While{T $\geq$ 1}{

         outcome = measure $\left((\cG\otimes \mathbb{I}_a)\cO_{d}\right)^T\ket{\psi}$ ; // outcome will be an index

         \If{$candidate$.\text{\upshape length} $>$ $\|f_{BDDP(L)}($\text{\upshape outcome}$)\|$}{

         $candidate$.length =  $\|f_{BDDP(L)}($outcome$)\|$ ;

         $candidate$.index = outcome;

         }

         T =T/2;
	}
}

\eIf{candidate.$\text{\upshape length} >  d'$}{

    \Return $\inf$;

}{

    \Return $candidate$.legnth;

}

\caption{Quantum algorithm for solving $d$-SVP, \qenump($L(B),d',\kappa$)}

\label{QdSVPsolver}

\end{algorithm}

The complete algorithm of $\qenump$ is given in Algorithm~\ref{QdSVPsolver} and we explain it as follows.
Given a lattice $L\subset \mathbb{R}^n$, we first randomly choose an index  $i\in \mathbb{Z}^n_3$, and  let $d'=f_{BDDP_L}(i)$. Then we run $\cO_d$ with $d=d'/\lambda_1$  as in Fig.~\ref{fig:UnknownGroverSearch}.
Suppose the set of indices  that are marked by $\cO_d$ is $S\subset \mathbb{Z}^n_3$ such that  $\forall i \in S,$ $\lambda_1\leq f_{BDDP_L}(i)\leq d\lambda_1$.
 An index  from $S$ is  uniformly chosen as a candidate.
  Then $\qenump$ is repeated for a total of $\kappa$ times. For each time,  the candidate is updated if we have a new candidate corresponding to a vector of shorter length.
At the end,
  we have
\begin{align}
Pr[\small\mbox{output } i\in S \mbox{  AND  } f_{BDDP(L)}(i) \mbox{ is shorter than half of } f_{BDDP(L)}(S)  \normalsize] \geq 1-\frac{1}{2^\kappa}, \label{eq:min_find}
\end{align}
where $f_{BDDP(L)}(S)$ means the collection of $f_{BDDP(L)}(x)$ for all $x \in S$.

Now we are ready to solve exact SVP.
 We will use  $\qenump$ as a subroutine in our quantum SVP solver $\qsvp$ as shown in Algorithm~\ref{QSVPsolver}.
 $\qenump$  will be executed with an updated smaller $d'$ to find an index corresponding to a shorter vector.
 The process continues until that no smaller $d'$ is found. Then a shortest vector will be found with high probability.
Consequently,  our main theorem is as follows.

\begin{algorithm}[h]

\SetAlgoLined

\SetKwData{Left}{left}\SetKwData{This}{this}\SetKwData{Up}{up}

\SetKwFunction{Union}{Union}\SetKwFunction{FindCompress}{FindCompress}

\SetKwInOut{Input}{input}\SetKwInOut{Output}{output}

\Input{lattice $L(\bB)$, confidence parameter $\kappa$}

\Output{shortest vector length $sv$}

Initialize: $d'$ = $f_{BDDP(L)}(i)$ for some $i\in \mathbb{Z}^n_3$;

Initialize: $timer = 0$;

\While{true}{

    \eIf{ $\text{\upshape \qenump}$($L(\bB), d',\kappa$) != $\inf$  $\text{\upshape and}$ $timer \leq \kappa\cdot n\log_23$}{

        $d'$ = $\left\|f_{BDDP(L)}\left(\qenump(L(\bB),d',\kappa)\right)\right\|$;

    }{

        \Return $d'$;

    }

}

\caption{Quantum algorithm for solving SVP, \qsvp($L$,$\kappa$)}

\label{QSVPsolver}

\end{algorithm}

\begin{theorem}
There is a quantum algorithm that solves SVP with probability $1-e^{\Omega(n)}$ using $O(ne^{(\beta^2/2e+o(1))n}\cdot 3^{n/2})=2^{1.2553n+o(n)}$ elementary quantum gates,   classic space $2^{n/2+o(n)}$ and poly$(n)$ qubits, where $\beta=2^{0.401}$.
\end{theorem}

\begin{proof}
{
Suppose we have a lattice $L\subset \mathbb{R}^n$. By Algorithm \ref{C_SVPalgorithm} and Theorem \ref{C_SVPsolver}, we know there exists a shortest vector $f_{BDDP(L)}(s)$ for some $s \in \mathbb{Z}^n_3$. In $\qsvp$, initially we choose an index $s\in \mathbb{Z}^n_3$ at uniform to compute $d=f_{BDDP(L)}(s)$ and use it to run the corresponding algorithm $\cO_{d}$, where $d\lambda_1(L)=d'$.
Once we have $\cO_{d}$ then we can solve $d$-SVP with $O(3^{n/2})$ queries to $\cO_{d}$ in the subroutine $\qenump$. Then we recursively update $d'$ in $\cO_{d}$ gate.
By Eq.~(\ref{eq:min_find}), for each $d'$, the solution set shrinks at least by half with probability $1-2^{-\kappa}$. Therefore  $d'$ is updated for $n\log_23$ times, and the probability that $\qsvp$ finds an index corresponding to $\lambda_1(L)$ is $(1-2^{-\kappa})^{n\log_23}>\frac{1}{2}$ by choosing $\kappa=\Omega(n)$.

The total complexity of $\qsvp$ is therefore $O(3^{n/2})$ times the complexity of $\cO_{d'}$. The complexity of $\cO_d$ is the sum of the complexity of $U_{BDDP(L)}$ and $\Filter_d$.
By  Corollary \ref{cor:BDDP}, Corollary \ref{smoothingDGS}, we can build a $0.39$-BDDP solver in time $O(2^{n/2})$ by using $O(e^{(\beta^2/2e+o(1))n})$ Toffoli gates and  $2^{n/2+o(n)}$ classic space. Also $\Filter_d$ can be built by using only $O(1)$ Toffili gates and $O(1)$ $X$ gates (see Fig. \ref{fig:Filter}). Hence the overall complexity of $\cO_d$   is $O(e^{(\beta^2/2e+o(1))n})$ Toffili gates and $O(1)$ $X$ gates. As a result, we can solve SVP with probability $1-e^{\Omega(n)}$ by $\qsvp$ with $2^{1.2553n+o(n)}$ Toffoli gates, $2^{n/2+o(n)}$  classic space, and poly$(n)$ qubits}
\end{proof}

\subsection{Filter$_d$}\label{sec:gap_SVP}

Now we discuss how to construct the circuit for Filter$_d$ for some $d>1$.
Let  $L\subset \mathbb{R}^n$ be a lattice.  Let $ c= d\lambda_1$.

We will use a binary representation for $\|f_{BDDP_L}(i)\|$ for all $i\in \mathbb{Z}^n_3$.
Suppose  $c$  is represented by $\ket{0}^{\otimes (l-1)}\ket{1}\ket{0}^{\otimes k}$.
Then for $v<c$, $v$ is represented by $\ket{0}^{\otimes l}\ket{b}$ for some  $b\in\mathbb{Z}_2^k$,
and for $v>c$, $v$ is represented by $\ket{a}\ket{b'}$ for some  $b'\in\mathbb{Z}_2^k$ and $a\in\mathbb{Z}_2^l\setminus\{0\}$.
So $l$ and $k$ determine the precision of this representation.

We want $\Filter_d$ to mark a quantum state $\ket{a'}$ with a phase $-1$ if it corresponds to value $a$ such that $\lambda_1\leq a <c$.
The resulting circuit is shown in Fig.~\ref{fig:dSVP}, which can be   efficiently constructed by using $O(l+k)$ Toffoli gates and $O(l)$ $X$ gates.
Note that the $l$-qubit control-NOT gate applies an $X$ to the target qubit when the $l$ control qubits are all $\ket{1}$.
Similarly,  the $l+k$-qubit control-control-NOT gate applies an $X$ to the target qubit when the $l+k$ control qubits are all $\ket{0}$.

$\Filter_d$ takes only $O(l+k)$ Toffoli gates and $X$ gates to construct.
  the dimension of $n$ will asymptotically increase to infinite and we can consider $l$ and $k$ constant when comparing them to the dimension of $n$. Hence we only need $O(1)$ Toffoli gates to construct $\Filter_d$.

\begin{figure*}[h]
\centering
\[
\Qcircuit @C=1em @R=.7em {
    \lstick{\ket{v_{\text{most}}}}  & /^l \qw & \gate{X^{\otimes l}} & \ctrl{2}    & \qw      & \qw       & \qw        & \ctrl{2}      & \gate{X^{\otimes l}} & \ctrlo{1}  &  \rstick{\ket{v'_{\text{most}}}}   \qw  \\
    \lstick{\ket{v_{\text{least}}}}  & /^k \qw & \qw                  & \qw         & \qw      & \qw       & \qw        & \qw            & \qw                  & \ctrlo{2}  &  \rstick{\ket{v'_{\text{least}}}}   \qw  \\
    \lstick{\ket{0}}  & \qw     & \qw                  & \gate{X}      & \gate{X} & \ctrl{1}  & \gate{X}   & \gate{X}         & \qw                  & \qw        &  \rstick{\ket{0}}    \qw  \\
    \lstick{\ket{-}}  & \qw     &  \gate{X}            & \qw         & \qw      & \gate{X}  & \qw        & \qw            & \qw                  & \gate{X}   &  \rstick{\ket{-}}    \qw
}\]
\vspace{0.5cm}
\fcaption{Filter$_d$, note $\ket{v_{\text{most}}}$ ($\ket{v_{\text{least}}}$) denotes the $l$-most ($k$-least) significant qubits.}
\label{fig:Filter}
\end{figure*}


\section{Conclusion and Open Problems}
Since the  Rivest-Shamir-Adleman (RSA) cryptosystem will be insecure against quantum attack in the future, it is desired to find candidates of post-quantum cryptosystems.
Several cryptographic tools are constructed from lattice problems such as SVP or CVP, which are believed to be quantum resistant.
In this paper, we proposed  classical and quantum algorithms for solving SVP with less space complexity.
   We constructed the enumeration algorithm $\enump$, which leads to  a classical SVP solver that runs in time $2^{2.0478n+o(n)}$ with space $2^{n/2+o(n)}$.
The classical SVP solver can be adapted to a quantum one that runs in time $2^{1.2553n+o(n)}$ and requires classical space $2^{n/2+o(n)}$ and  only poly$(n)$ qubits.

One would like to know whether a space-time tradeoff is possible in our scheme.
More explicitly, can we use a little more space so that the time complexity can be reduced? Unfortunately, the answer is no.
As mentioned in Section~\ref{sec:BDDP}, an $\alpha$-BDD oracle from discrete Gaussian sampling can  have $\alpha$ as high as $0.5-o(1)$
and   the search space is at least $3^n$, which implies the query complexity for Grover search is at least $O(3^{0.5n})$.
In addition, there is no space-time tradeoff in Theorem~\ref{hDGS} to generate discrete Gaussian distribution.
If we choose a smaller $\alpha$, instead of $\alpha=1/3$, the BDD oracle, whose time complexity is dominated by the preparation of discrete Gaussian distribution,
still needs  $O(2^{n/2})$ time.


Nevertheless, it is still possible to further reduce the time complexity of our Algorithms~\ref{C_SVPalgorithm} and \ref{QSVPsolver} by providing a more efficient method for discrete Gaussian sampling. Recall that in Corollary~\ref{smoothingDGS},  $O(2^{n/2})$ vectors are sampled in time $O(2^{n/2})$ from $D_{L^*,\eta_\epsilon(L^*)}$ with $\epsilon = e^{-(2^{0.802}/e+o(1))n}$ and then used  to construct a $0.391$-BDDP solver.
  Clearly we only need a $1/3$-BDDP oracle to solve SVP by using $\enump$ (Algorithm~\ref{C_SVPalgorithm}).
  To construct a $1/3$-BDDP oracle, we only need   $2^{0.1604n+o(n)}$ vectors from $D_{L^*,\eta_\epsilon(L^*)}$ with $\epsilon = 2^{-0.3208n}$
  by Theorem~\ref{CVPtoDGS} and Theorem~\ref{hDGS}.
   However,    $\eta_\epsilon(L^*)$ is not greater than $\sigma(L)$  and Theorem~\ref{hDGS}  cannot be applied.
   Once we can sample $2^{0.1604n}$ vectors from $D_{L^*,\eta_\epsilon(L^*)}$ with $\epsilon = 2^{-0.3208n}$ in $O(2^{0.9594n})$, then our Algorithm~\ref{C_SVPalgorithm} $\enump$ can be improved to  solve SVP in time $2^{1.7584n+o(n)}$ and our Algorithm~\ref{QSVPsolver} $\qsvp$ can potentially find the shortest vector in time $2^{0.9594n+o(n)}$.

As for heuristic algorithms for SVP, 
Laarhoven \emph{et al.} conducted a systematic study on using quantum search to speed up several existing heuristic algorithms and they had a heuristic quantum SVP solver with time and space complexity both $2^{0.265n}$. 
We wonder if other quantum algorithms or techniques can be exploited for heuristic algorithms as well.
However, it is not clear whether we can add any heuristic assumption in our algorithms.

\nonumsection{References}

\bibliographystyle{IEEEtran}

\end{document}